%% file: main.tex
\documentclass[journal]{IEEEtran}
\IEEEoverridecommandlockouts

\usepackage{amsmath}
\usepackage{mathtools}
\usepackage{amsmath}
\usepackage{amssymb}
\usepackage{amsthm}
\usepackage{cite}
\usepackage{caption}
\usepackage{graphicx}
\usepackage{caption}
\usepackage{subfig}
\usepackage{subfloat}
\usepackage[bottom]{footmisc}

\usepackage[table]{xcolor}

\usepackage[utf8]{inputenc}
\usepackage{fourier}
\usepackage{array}
\usepackage{makecell}
\usepackage{xr}

\usepackage[ruled,linesnumbered]{algorithm2e}

\DeclareMathOperator{\atan}{tan^{-1}}
\DeclareMathOperator{\sumn}{\sum_{n=1}^\infty}
\DeclareMathOperator{\onen}{(-1)^{n+1}}

\newtheorem{theorem}{Theorem}%[section]
\newtheorem{corollary}{Corollary}%[theorem]
\begin{document}

%\title{Information transmission bounds in mobility enabled communication channels}
\title{Information transmission bounds between moving terminals}

\author{Omar J. Faqir, Eric C. Kerrigan, Deniz G\"und\"uz% <-this % stops a space
    \thanks{\textcopyright 2020 IEEE.  Personal use of this material is permitted.  Permission from IEEE must be obtained for all other uses, in any current or future media, including reprinting/republishing this material for advertising or promotional purposes, creating new collective works, for resale or redistribution to servers or lists, or reuse of any copyrighted component of this work in other works.}   
	\thanks{The support of the EPSRC Centre for Doctoral Training in High Performance Embedded and Distributed Systems  (HiPEDS, Grant Reference EP/L016796/1) is gratefully acknowledged. 
	This work was also supported in part by the European Research Council (ERC) through Starting Grant BEACON (agreement no. 677854).}% <-this % stops a space
	\thanks{O. J. Faqir and Deniz G\"und\"uz are with the Department of Electrical \& Electronics Engineering, Imperial College London, SW7~2AZ, U.K. {\tt\small ojf12@ic.ac.uk}, {\tt\small d.gunduz@ic.ac.uk}}%
	\thanks{Eric C. Kerrigan is with the Department of Electrical \& Electronic Engineering and Department of Aeronautics, Imperial College London, London SW7~2AZ, U.K. {\tt\small e.kerrigan@imperial.ac.uk}}%
}

% \author{Author 1, Author 2, Author 3% <-this % stops  space
% }%

\maketitle

\begin{abstract}
	In networks of mobile autonomous agents, e.g.\ for data acquisition, we may wish to maximize data transfer or to reliably transfer a minimum amount of data, subject to quality of service or energy constraints.
	These requirements can be guaranteed through both offline node design/specifications and online trajectory/communications design.
	Regardless of the distance between them, for a stationary point-to-point transmitter-receiver pair communicating across a single link under average power constraints, the total data transfer is unbounded as time tends to infinity.
	In contrast, we show that if the transmitter/receiver is moving at any constant speed away from each other, then the maximum transmittable data is bounded. Although general closed-form expressions as a function of communication and mobility profile parameters  do not yet exist, we provide closed-form expressions for particular cases, such as ideal free space path loss.
	Under more general scenarios we instead give lower bounds on the total transmittable information across a single link between mobile nodes.
\end{abstract}

\section{Introduction}
As autonomous agents become commonplace, the challenging task of developing control policies for communication between agents must be addressed.
For a stationary point-to-point transmitter-receiver pair under an average power constraint, the total transmittable data is unbounded as time tends to infinity, regardless of distance. In contrast, we show that if the nodes move away from each other at a constant speed, then the maximum transmittable data is bounded even as time tends to infinity. 
Without loss of generality we consider one stationary and one mobile node.

This bound is of relevance for aerial communications, where channel gain is dominated by path loss exponents and typically high speeds result in rapid growth in distances. Examples include unmanned aerial vehicles (UAVs) or spacecraft.
We consider dynamic deployment scenarios (e.g.\   \cite{faqir2018energy,zeng2017energy,wu2019fundamental,zeng2016throughput}) where UAVs, spacecraft, or other agents are transmitting while moving, and are interested in characterising the total data transfer as a function of mobility dynamics and design parameters, such as transmit power.
We do not consider transmitter power control as a function of distance as this is challenging to implement in practice, especially at higher speeds.
Linear constant speed trajectories, as considered in \cite{faqir2018energy,zeng2017energy}, have practical importance for fixed-wing aircraft, since minimum propulsion energy operation occurs at a constant velocity, and result in structured time-varying communication channels. 
For dense UAV deployment scenarios with random mobility, works such as \cite{yuan2018capacity} determine capacity bounds and outage probabilities.
Even for a simple linear trajectory in 3D space, the expression for total transmittable data formulated in Section~\ref{sec:problemDef} is non-trivial to evaluate due to the structure of the argument of the logarithm. We instead derive closed-form expressions for specific cases. In Section~\ref{sec:OneDimension} we derive a bound for total transmittable data in 1D and show the bound to be tight for reasonable parameter values. In Section~\ref{sec:alpha2} we consider arbitrary linear agent trajectories deriving an expression for total data transmission, assuming ideal free space path loss.

\section{Problem Definition}\label{sec:problemDef}
The Shannon capacity is an information-theoretic limit on the achievable communication rate for reliable transmission across a noisy channel. The corresponding limit on the total amount of transmittable data over time $T$ is
\begin{equation}\label{eq:Shannon_integral}
D_T := B\int_0^T \ln\left(1 + \operatorname{SNR}(t)\right)\mathrm{d}t,
\end{equation}
where $B$ is the bandwidth and $\operatorname{SNR}(\cdot)$ the time-varying signal-to-noise ratio at the receiver. $D_T$ is measured in nats (natural units of information). Neglecting fast-fading dynamics, the channel gain is dominated by line-of-sight fading, typical of aerial channels \cite{faqir2018energy}. Then, the $\operatorname{SNR}$ is %modelled as
\begin{equation} \label{eq:SNR}
\operatorname{SNR}(P,t) = \frac{P G}{\sigma^2}\left(\frac{d_0}{d(t)}\right)^\alpha,
\end{equation}
where $\sigma^2$ is the receiver noise power, $\alpha\geq2$ (but typically not much greater than $2$) is the path loss exponent, and $d(\cdot)$ is the time-varying distance between transmitter and receiver, satisfying $d(t)>d_0,\forall t\in\mathbb{R}^+,$ for reference distance~$d_0$. $P$ is the average transmission power in Watts and $G$ is a unitless antenna parameter representing gain and path loss at the distance $d_0$. For simplicity we define $S:=PG\sigma^{-2}$.

Consider wireless transmission between a stationary node at the origin and a mobile transmitter at constant longitudinal displacement $z_0\geq 0$. The transmitter follows a linear trajectory at constant speed $v$. At time $t \in [0,T]$ the lateral position of the transmitter relative to the receiver is
\begin{equation} \label{eq:UAV_position}
x(t) = x_0 + vt,
\end{equation}
and $x_0$ the initial lateral position. The distance is
\begin{equation} \label{eq:distance}
d(t) = \sqrt{z_0^2 + (x_0 + vt)^2}.
\end{equation}
Simplifying the distance operator through removing the $\sqrt{(\cdot)}$, either by assuming $z_0=0$ (Section~\ref{sec:OneDimension}) or setting $\alpha=2$ (Section~\ref{sec:alpha2}), allows us to determine closed-form expressions for the total amount of data that can be reliably communicated in these special cases.

\section{Characterization for general $\alpha$ in 1D}\label{sec:OneDimension}
 
We find the total data transfer in a one-dimensional setting, where $z_0=0$, and hence $d(t)=x_0+vt$, as a function of $B$, $v$, $S$ and $\alpha$. The restriction $z_0=0$ simplifies the analysis and arises in the context of a mobile UAV communicating with a static aerial base station at the same altitude.

\begin{theorem} \label{th:one_dimension}
 	The total transmittable data as the interval length $T\rightarrow \infty$ for %a node with
 	$z_0=0$, $x_0= d_0$ and $v>0$, i.e.
 	\begin{equation} \label{eq:int_inf_1D}
 	D_{\infty,1} := \lim_{T\rightarrow \infty} B \int_0^T \ln\left(1 + S\left(\frac{d_0}{d_0 + vt}\right)^\alpha\right) \mathrm{d}t,
 	\end{equation}
 	is finite and given by
	\begin{equation} \label{eq:D_inf_1D}
	\begin{split}
	D_{\infty,1} = \frac{B d_0}{v} \Bigg(& \pi\sqrt[\alpha]{S} \csc\left(\frac{\pi}{\alpha}\right) - \ln(1+S) + \\
	&\alpha \sum_{n=1}^\infty \left\{\frac{(-1)^{n+1}}{S^n(\alpha n + 1)}\right\} - \alpha \Bigg),
	\end{split}
	\end{equation}
	where the sum is positive and convergent for $S>1$.
\end{theorem}
See Appendix~\ref{app:A} for the proof. Restricting $S>1$ implies that the transmit signal power, before undergoing any channel attenuation, is greater than the receiver noise power. In practice we find it reasonable to assume $S\gg1$.
Since the sum is monotonically decreasing in both $S$ and $\alpha$, we may determine a lower bound $\tilde{D}_{\infty,1}$ on $D_{\infty,1}$.

\begin{corollary} \label{co:approximation}
A lower bound for $D_{\infty,1}$ is given by
	\begin{equation} \label{eq:D_inf_1D_approx}
	\begin{split}
	\tilde{D}_{\infty,1} := \frac{B d_0}{v} \left(\pi\sqrt[\alpha]{S} \csc\left(\frac{\pi}{\alpha}\right) - \ln(1+S) - \alpha \right),
	\end{split}
	\end{equation}
	where the error $e_{\infty,1}$ between $D_{\infty,1}$ and the bound $\tilde{D}_{\infty,1}$, 
	\begin{equation} \label{eq:ApproxError}
		e_{\infty,1} := D_{\infty,1} - \tilde{D}_{\infty,1} = \frac{B d_0}{v}\sum_{n=1}^\infty \left\{\frac{(-1)^{n+1}}{S^n(\alpha n + 1)}\right\},
	\end{equation}
	 is bounded as
	\begin{equation}  \label{eq:D_inf_1D_approx_error}
	e_{\infty,1} \leq \frac{B d_0 \alpha}{v} \left( 1 - \sqrt{S}\atan\left(\frac{1}{\sqrt{S}}\right)\right) \leq \frac{B d_0 \alpha}{v} \left( 1 - \frac{\pi}{4}\right).
	\end{equation}
	Furthermore, $e_{\infty,1} \rightarrow 0$ as $S \rightarrow \infty$.
\end{corollary}

See Appendix~\ref{app:B} for a proof. Fig.~\ref{fig:1DApproxError} compares $\tilde{D}_{\infty,1}$ and $D_{\infty,1}$ for the parameters shown in Table~\ref{tab:Params} and a range of~$S$\footnote{For the sake of relatability all data is plotted in MB, not nats.}. Setting $P=1$\,mW results in $S=10^5$. The corresponding error is shown in Fig.~\ref{fig:1DApproxError_2} as $e_{\infty,1} = 3.33\times 10^{-6}$\,MB$=0.278$\,nats. Fig.~\ref{fig:1DApproxError_2} also shows that $e_{\infty,1}=\frac{B d_0 \alpha}{v} \left( 1 - \frac{\pi}{4}\right) = 1.5\times10^{-3}$\,MB when $S=1$, but decreases rapidly as $S$ increases. 

\input{Table1.tex}
% \begin{table}[tb]
% 	\centering
% 	\caption[TODO]{Table of \textit{default} simulation parameters. Graph axes and labels indicate where parameters differ from defaults.}
% 	\begin{tabular}{|l|l|l|l|l|l|l|}
% 		\hline
% 		$B$(Hz) & $\sigma^{2}$(W) & $d_0$(m) & $G$ & $P$(W) & $v\text{(ms}^{-1}$)& $\alpha$ \\
% 		\hline
% 		$10^5$ & $10^{-8}$ & $1$ & $1$ & $1$ & $5$ & $2$ \\
% 		\hline
% 	\end{tabular}
% 	\label{tab:Params}
% \end{table}
\begin{figure}[t]
    \vspace{-0.6cm}
	\centering
	\subfloat[Comparison of $D_{\infty,1}$ and $\tilde{D}_{\infty,1}$]{%
		\includegraphics[width=0.5\columnwidth]{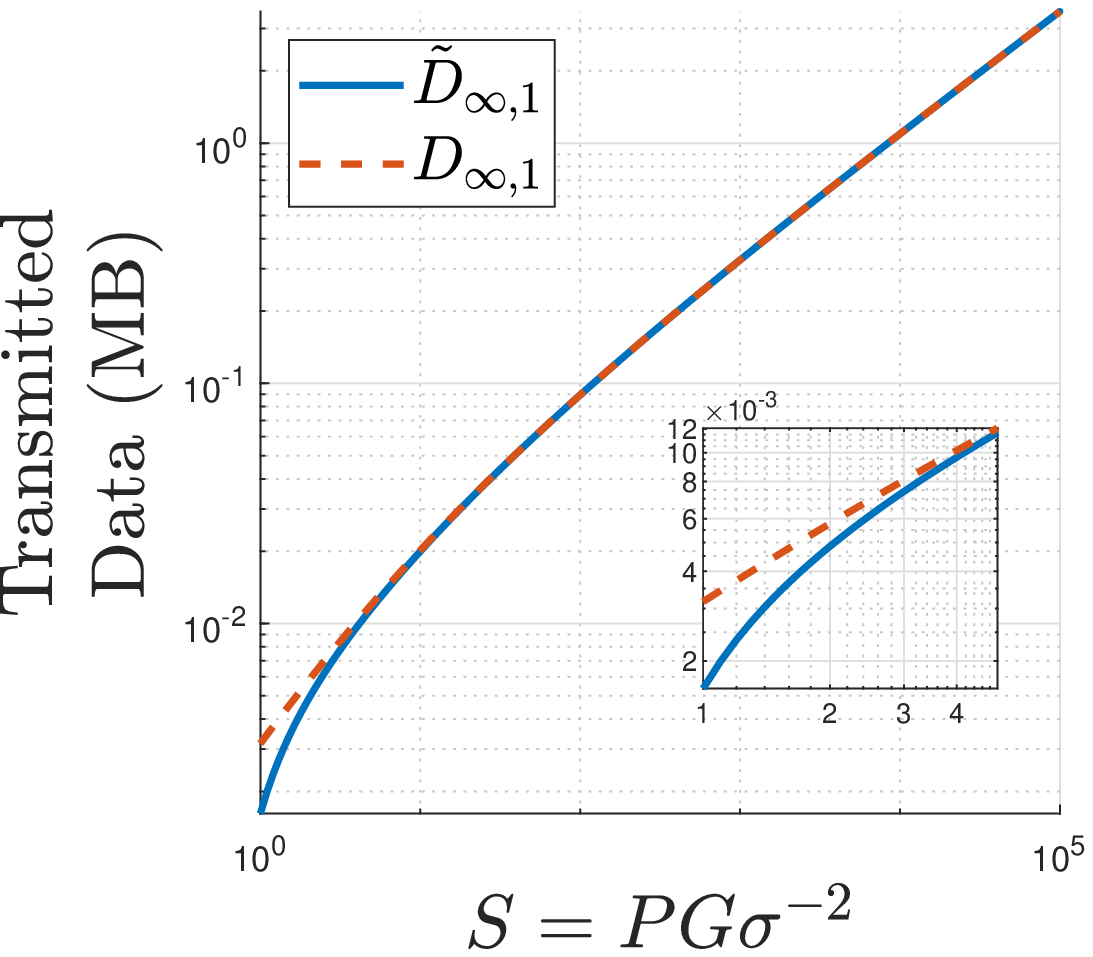}%
		\label{fig:1DApproxError_1}%
	}
	\subfloat[Comparison of $e_{\infty,1}$ with bound]{%
		\includegraphics[width=0.5\columnwidth]{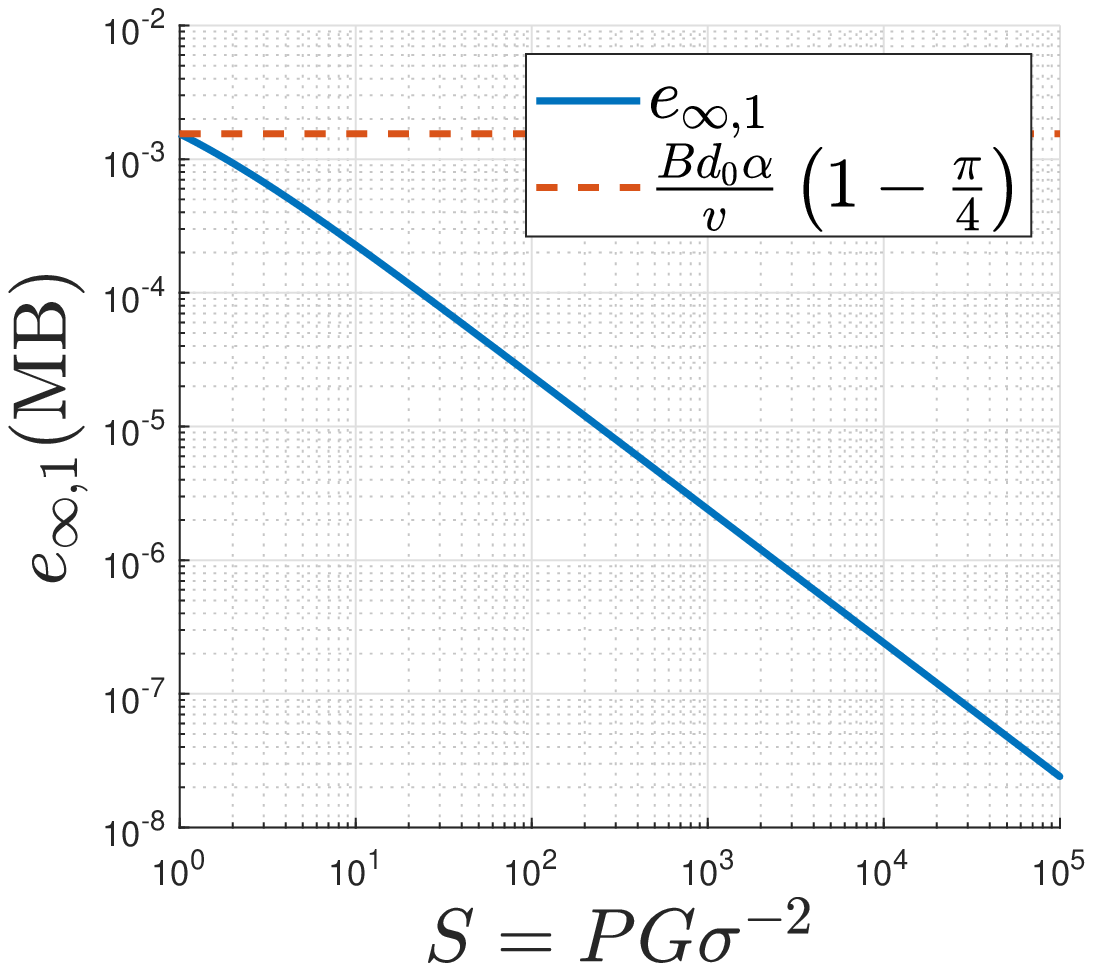}%
		\label{fig:1DApproxError_2}%
	}
	\caption{Comparison of $D_{\infty,1}$, $\tilde{D}_{\infty,1}$.}
	\label{fig:1DApproxError}
\end{figure}

Fig.~\ref{fig:Power_Velocity_Figure} shows $D_{\infty,1}$ as a function of speed $v\in[1,100]$ and power $P\in[10^{-3},100]$, calculated using \eqref{eq:D_inf_1D}. We approximate the infinite sum by the first $100$ terms.
% As suggested by Fig.~\ref{fig:1DApproxError}, there is agreement between $D_{\infty,1}$ and $\tilde{D}_{\infty,1}$ for the considered range of $(P,v)$.
Since $e_{\infty,1}$ is monotonically decreasing in $P$ and $v$, the maximum approximation error corresponds to data point $(P,v)=(10^{-3},1)$ and is $e_{\infty,1}=0.333$\,nats.
Although controlling node mobility is outside of the scope of this paper, Fig.~\ref{fig:Power_Velocity_Figure}  shows the set $\{(P,v) \mid D_{\infty,1} = M\}$ of admissible powers and velocities that successfully transmit $M$\,bits of data. $(P,v)$ may then be chosen from this set to achieve alternative operating points, such as energy efficiency \cite{faqir2018energy}. Since many transmitters operate at a single, or finite set of power levels,  $v$ may often be found as an explicit function of $(P,M)$.
% The set of admissible control variables
% is f(P; v) j D1 = Mg, from which (P; v) may then be
% chosen to achieve further control objectives, such as energy
% efficiency [1]. Furthermore, many transmitters only operate at
% a single, or finite set of power levels, in which case v may be
% found as an explicit function of (P; M). Importantly, when
% P is not a control variable, constant speed flight is often the
% most energy efficient.
\begin{figure}[b]
	\centering
		\includegraphics[width=0.9\columnwidth]{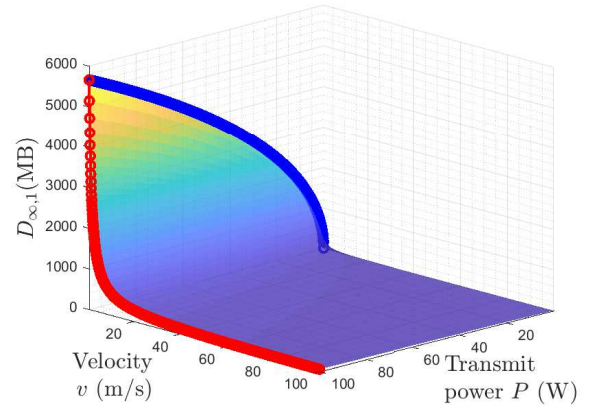}
	\caption[TODO]{$D_{\infty,1}$ as a function of speed $v$ and transmit power $P$  in the 1D case ($z_0=0$) for parameters from Table~\ref{tab:Params}.}
	\label{fig:Power_Velocity_Figure}
\end{figure}
For comparison, the red and blue lines show $D_{\infty,1}$ calculated from the Matlab  quadrature function \texttt{integral} of \eqref{eq:Shannon_integral} for infinite $T$.
% The quality of such a quadrature approximation depends on the length of the horizon $T,$ where larger transmit powers result in a slower decay of \eqref{eq:SNR} and mandate a larger $T$.
Fig.~\ref{fig:FiniteTimePlot} shows $D_T/D_{\infty,1}$ as a function of both time and position for a range of speeds. $D_T$ converges faster to $D_{\infty,1}$ with increasing $v$. Even for a moderately slow speed of $v=5$\,m/s, over $80\%$ of the data is transmitted within one hour.
%This also relates to realistic scenarios where quality of service requirements on the SNR are tighter than the theoretical bound $D_{\infty,1}$.
For the AWGN case here, with fixed power and speed, any SNR type constraint can be directly converted to a maximum allowable distance between the transmitter and receiver.
Similarly, if constraints on node endurance (e.g.\ finite energy) are functions of velocity, one may quantitatively investigate how limiting endurance is, compared to the bound $D_{\infty,1}$. Fig.~\ref{fig:FiniteTimePlot} shows a decreasing marginal gain in transmitted data as endurance increases.

\begin{figure}[tb]
	\centering
	\includegraphics[width=\columnwidth]{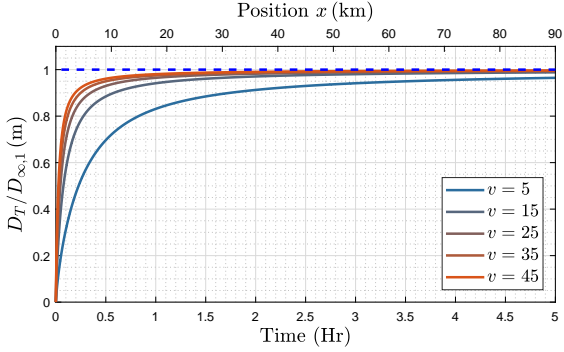}
	\caption[TODO]{Finite-time data $D_T$ as a proportion of the bound $D_{\infty,1}$ in the 1-D case for a range of speeds $v$ in m/s.}
	\label{fig:FiniteTimePlot}
\end{figure}

\section{Characterization for $\alpha=2$} \label{sec:alpha2}

We now consider the specific, but not unrealistic case of $\alpha=2$, corresponding to ideal free space. The maximum transmittable data over time $T$ simplifies to
\begin{equation}\label{eq:DFiniteTime}
D_T = B \int_0^T \ln\left(1 + \frac{S d_0^2}{z_0^2 + x(t)^2}\right)\mathrm{d}t.
\end{equation}
%for which we may find a closed-form expression.

\begin{theorem} \label{th:alpha2}
	The total transmittable data as  $T\rightarrow\infty$ for a UAV with starting position $x_0\geq d_0$, $v>0$ and $\alpha=2$, i.e.
	\begin{equation} \label{eq:int_inf_alpha2}
	D_{\infty,2} :=  \lim_{T\rightarrow \infty} B \int_0^T \ln\left(1 + \frac{S d_0^2}{z_0^2 + x(t)^2}\right)\mathrm{d}t,
	\end{equation}
	for  trajectory \eqref{eq:UAV_position} is finite and given by
	\begin{equation}
	\begin{split} \label{eq:D_infty_2}
		%D_{\infty,2} &= \frac{B}{v} \Bigg[ -x_0\ln \left( 1 + \frac{Sd_0^2}{x_0^2 + z_0^2}\right)  + \pi\left(\sqrt{Sd_0^2+z_0^2}-z_0\right)\\
		%&-2 \left(\sqrt{z_0^2 + Sd_0^2}\atan\left(\frac{x_0}{\sqrt{z_0^2 + Sd_0^2}}\right) -x_0\atan\left(\frac{x_0}{z_0}\right)\right)
		%\Bigg].
		D_{\infty,2} &= \frac{B}{v} \Bigg[ -x_0\ln \left( 1 + S\right)  + \pi\left(\sqrt{\epsilon}-z_0\right)\\
		&-2 \left(\sqrt{\epsilon}\atan\left(\frac{x_0}{\sqrt{\epsilon}}\right) -x_0\atan\left(\frac{x_0}{z_0}\right)\right)
		\Bigg].
	\end{split}
	\end{equation}
	where $\epsilon := z_0^2 + Sd_0^2$.
\end{theorem}

The proof is in Appendix~\ref{app:C}. Fig.~\ref{fig:z0} shows how $D_{\infty,2}$ varies with displacement $z_0$.
In the straightforward case of $\alpha=2, \ z_0=0$ a simpler expression
% for total transmittable data
follows.

\begin{figure}[tb]
	\centering
	\includegraphics[width=\columnwidth]{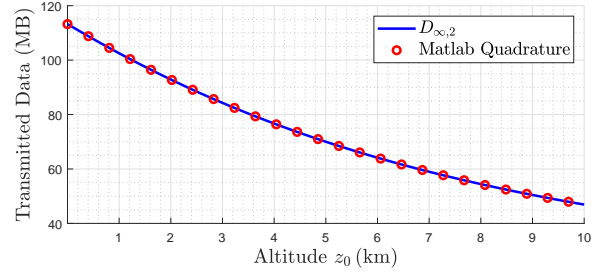}
	\caption[TODO]{$D_{\infty,2}$ as a function of altitude $z_0$.}
	\label{fig:z0}
\end{figure}

\begin{corollary}
	In the 1-D case ($z_0=0$), with $x_0 = d_0$ and $v>0$,
	\begin{equation}
	D_{\infty,3} :=  \lim_{T\rightarrow \infty} B \int_0^T \ln\left(1 + \frac{S d_0^2}{x(t)^2}\right)\mathrm{d}t,
	\end{equation}
	 the maximal transmittable data as $T\rightarrow \infty$ is
	\begin{equation}
	D_{\infty,3} = \frac{Bd_0}{v}\Bigg[ 2\sqrt{S}\left(\frac{\pi}{2} - \atan\left(\frac{1}{\sqrt{S}}\right)\right) -
	\ln\left(1+S\right)\Bigg].
	\end{equation}
\end{corollary}
\begin{proof}
	The proof follows by setting $z_0=0,\ x_0=d_0$ in \eqref{eq:D_infty_2} or by setting $\alpha=2$ in \eqref{eq:D_inf_1D} and noting (e.g. using \cite[p. 422]{knopp1990theory})
	\begin{equation}
	\sum_{n=1}^\infty \left\{\frac{(-1)^{n+1}}{S^n(2n+1)}\right\} = 1 - \sqrt{S}\atan\left(\frac{1}{\sqrt{S}}\right).
	\end{equation}
\end{proof}

\section{Conclusions}
 We have considered the amount of data that can be reliably transmitted between a stationary and a mobile agent on a straight trajectory. We have shown that, as opposed to static channels, the maximum transmittable data is bounded, even as time tends to infinity. For certain special cases of path loss and mobility profiles, we have been able to derive closed-form expressions for these bounds. These bounds have been verified through simulations and may be used for design/specification of mobile communication channels subject to data transmission constraints.

\appendices
\section{Proof of Theorem~\ref{th:one_dimension}} \label{app:A}
%\begin{proof}
	Note, e.g. from \cite[p. 212]{knopp1990theory}, the identity
	\begin{equation} \label{eq:identity_log1}
	\ln(1+z) =
	\begin{cases}
	\sum_{n=1}^\infty \frac{(-1)^{n+1}z^n}{n},& \text{if } |z| < 1, z \in \mathbb{R},\\
	\ln(z) + \sum_{n=1}^\infty \frac{(-1)^{n+1}}{nz^n},& \text{if } |z| \geq 1,  z \in \mathbb{R},
	\end{cases}
	\end{equation}
	and
	\begin{equation} \label{eq:identity_csc}
	\text{csc}(z) = \frac{1}{z} + 2z \sum_{n=1}^\infty \frac{(-1)^n}{z^2 - (\pi n)^2} \quad ; \frac{z}{\pi} \not \in \mathbb{Z}, z \in \mathbb{R}.
	\end{equation}
	Starting with \eqref{eq:int_inf_1D}, we use the known affine trajectory $x(t) = x_0 + vt$ to change the variable of integration, resulting in
	\begin{equation*}
	D_{\infty,1} = \lim_{T\rightarrow \infty}\frac{B}{v} \underbrace{\int_{d_0}^{x_T} \ln\left(1 + S\left(\frac{d_0}{x}\right)^\alpha\right) \mathrm{d}x}_{:=D_a},
	\end{equation*}
	where $x_T := x(T)$. Define $\gamma := Sd_0^\alpha$. Ignoring the limits for now, and in accordance with identity \eqref{eq:identity_log1},
	\begin{align*}
	D_a &= \int_{d_0}^{\sqrt[\alpha]{\gamma}} \ln\left(1 + \frac{\gamma}{x^\alpha}\right)\mathrm{d}x +
	\int_{\sqrt[\alpha]{\gamma}}^{x_T} \ln\left(1 + \frac{\gamma}{x^\alpha}\right)\mathrm{d}x \\
	D_a &= \underbrace{\int_{d_0}^{\sqrt[\alpha]{\gamma}} \ln\left(\frac{\gamma}{x^\alpha}\right) + \sumn \frac{\onen}{n}\left(\frac{x^\alpha}{\gamma}\right)^n \mathrm{d}x}_{:=D_{a1}} \\ & +
	\underbrace{\int_{\sqrt[\alpha]{\gamma}}^{x_T} \sumn \frac{\onen}{n}\left(\frac{\gamma}{x^\alpha}\right)^n \mathrm{d}x}_{:= D_{a2}},
	\end{align*}
	where the domain of integration has been split at $\gamma x^{-\alpha}=1$. Considering just $D_{a1}$, we are able to switch the order of integration and summation, because the infinite sum is absolutely and uniformly convergent on the domain of integration, as may be shown through the Weierstrass M-test \cite[Theorem~7.10]{rudin1964principles}. Therefore,
	\begin{align*}
	D_{a1} &= \left[ x \left( \ln\left(\frac{\gamma}{x^\alpha}\right) + \alpha\right) +
	\sumn \frac{\onen}{n}\left(\frac{x^{\alpha n + 1}}{\gamma^n (\alpha n + 1)}\right) \right]_{x_0}^{\sqrt[\alpha]{S}} \\
	&= d_0 \left[ \alpha \sqrt[\alpha]{S} - \alpha - \ln(S) + \sumn \frac{\onen}{n(\alpha n + 1)}\left(\sqrt[\alpha]{S} - \frac{1}{S^n}\right)\right].
	\end{align*}
	The order of integration and summation may similarly be swapped in $D_{a2}$, resulting in
	\begin{align*}
	D_{a2} &= \sumn \frac{\onen}{n(1-\alpha n)} \left[x\left(\frac{\gamma}{x^\alpha}\right)^n\right]_{\sqrt[\alpha]{\gamma}}^{x_T} \\
	 &= \sumn \frac{\onen}{n(1-\alpha n)} \left(x_T\left(\frac{\gamma}{x_T^\alpha}\right)^n - \sqrt[\alpha]{S}d_0\right).
	\end{align*}
	The term $D_{a1}$ does not change with $x_T$, and the limit
	\begin{equation*}
	\lim_{x_T\rightarrow\infty} D_{a2} = -\sumn \frac{\onen}{n(1-\alpha n)} \left(\sqrt[\alpha]{S}d_0\right).
	\end{equation*}
	Therefore, the total transmittable data
	\begin{equation*}
	\begin{split}
	D_{\infty,1} &=  \frac{B d_0}{v} \Bigg[ \alpha \sqrt[\alpha]{S} - \alpha - \ln(S) \\& + \sumn \frac{\onen}{n(\alpha n + 1)}\left(\sqrt[\alpha]{S} - \frac{1}{S^n}\right) - \sqrt[\alpha]{S}\sumn \frac{\onen}{n(1-\alpha n)}\Bigg].
	\end{split}
	\end{equation*}
	Noting that,
	\begin{equation}
(n (\alpha n + 1))^{-1} - (n (1-\alpha n))^{-1}
	= -2\alpha(1-(\alpha n)^2)^{-1},
	\end{equation}
	we may rewrite the infinite sums in $D_{\infty,1}$, resulting in
	\begin{equation*}
	\begin{split}
	D_{\infty,1} &=  \frac{B d_0}{v} \Bigg[ \alpha \sqrt[\alpha]{S} - \alpha \underbrace{- \ln(S)  - \sumn\frac{\onen}{n S^n}}_{-\ln(1+S)} \\ & + \alpha \sumn \frac{\onen}{(\alpha n + 1)S^n} - \underbrace{2\alpha\sqrt[\alpha]{S} \sumn \frac{\onen}{1-(\alpha n)^2}}_{-\left(\pi\sqrt[\alpha]{S}\csc\left(\frac{\pi}{\alpha}\right)-\alpha\sqrt[\alpha]{S}\right)} \Bigg],
	\end{split}
	\end{equation*}
	where the transformation in the first line follows from \eqref{eq:identity_log1}. Restricting $S >1$, the transformation in the second line follows from \eqref{eq:identity_csc}. Since $\alpha\geq2$, the restriction of the domain of \eqref{eq:identity_csc} to $\frac{z}{\pi} \not \in \mathbb{Z}$ is not prohibitive. Finally,
		\begin{equation}
% 	\begin{split}
	D_{\infty,1} = \frac{B d_0}{v} \Bigg[ \pi\sqrt[\alpha]{S} \csc\left(\frac{\pi}{\alpha}\right) - \ln(1+S) + 
	\alpha \sumn \left\{\frac{\onen}{S^n(\alpha n + 1)}\right\} - \alpha \Bigg].
% 	\end{split}
	\end{equation}
%\end{proof}

\section{Proof of Corollary~\ref{co:approximation}}\label{app:B}
%\begin{proof}
	Simply ignoring the infinite sum results in \eqref{eq:D_inf_1D_approx}. The infinite sum in \eqref{eq:D_inf_1D} is bounded for $S>1$ if $\alpha\geq2$ and is monotonically decreasing in both $\alpha$ and $S$. Therefore,
	\begin{align*}
	\sumn \left\{\frac{\onen}{S^n(\alpha n + 1)}\right\} \leq  \sumn & \left\{\frac{\onen}{S^n(2 n + 1)}\right\}= \\
	&1 - \sqrt{S}\atan\left(\frac{1}{\sqrt{S}}\right) \leq 1 - \frac{\pi}{4},
	\end{align*}
	where the infinite sum is evaluated noting that \cite[p. 422]{knopp1990theory}
	\begin{align*}
	\atan(z) &= -\sum_{n=0}^\infty (-1)^{n+1}\frac{z^{2n+1}}{2n+1}, z \in \mathbb{R}.
	\end{align*}
	 Therefore,
	\begin{equation*}
		D_{\infty,1} - \tilde{D}_{\infty,1} \leq \frac{B d_0 \alpha}{v} \left( 1 - \sqrt{S}\atan\left(\frac{1}{\sqrt{S}}\right)\right) \leq \frac{B d_0 \alpha}{v} \left( 1 - \frac{\pi}{4}\right).
	\end{equation*}
	Since $\lim_{S\rightarrow \infty} \sqrt{S}\atan\left(\frac{1}{\sqrt{S}}\right)=1$, the error $e_{\infty,1}\rightarrow 0$ as $S\rightarrow \infty$.
%\end{proof}

\section{Proof of Theorem~\ref{th:alpha2}}\label{app:C}
%\begin{proof}
	Note that, using integration by parts,
	\begin{equation}\label{eq:identity_log2}
	\int_c^b \ln(z^2 + a)\mathrm{d}z= \left[z\left(\ln(z^2+a) - 2\right) + 2\sqrt{a}\atan\left(\frac{z}{\sqrt{b}}\right) \right]_c^b.
	\end{equation}
	We begin by rewriting \eqref{eq:int_inf_alpha2} as
	\begin{equation}
		D_{\infty,2} =  \lim_{x_T\rightarrow \infty} B \int_{x_0}^{x_T} \ln\left((z_0^2 + S d_0^2) + x^2\right) - \ln\left(z_0^2 + x^2\right)\mathrm{d}x.
	\end{equation}
	Integrating the first logarithm, using $\epsilon=z_0^2+Sd_0^2$ and \eqref{eq:identity_log2},
	\begin{align*}
 		\int_{x_0}^{x_T} \ln&\left((z_0^2 + S d_0^2) + x^2\right) \mathrm{d}x = 2(x_0-x_T) + x_t\ln(x_T^2+\epsilon) \\ & \ - x_0\ln(x_0^2+\epsilon) + 2\sqrt{\epsilon}\left(\atan\left(\frac{x_T}{\sqrt{\epsilon}}\right) -
 			\atan\left(\frac{x_0}{\sqrt{\epsilon}}\right)\right).
	\end{align*}
	Similarly, the second logarithm is integrated as
	\begin{align*}
	 \int_{x_0}^{x_T} \ln & \left(z_0^2 + x^2\right)\mathrm{d}x = 2(x_0-x_T) + x_T\ln(x_T^2+z_0^2)
	 \\ & - x_0\ln(x_0^2+z_0^2) + 2z_0\left(\atan\left(\frac{x_T}{z_0}\right) -
	\atan\left(\frac{x_0}{z_0}\right)\right).
	\end{align*}
	Therefore,
	\begin{align*}
	D_{\infty,2} =&  \lim_{x_T\rightarrow \infty}  \frac{B}{v} \Bigg[ x_T\ln\left(\frac{x_T^2+\epsilon}{x_T^2+z_0^2}\right) +
	x_0\ln\left(\frac{x_0^2+z_0^2}{x_0^2+\epsilon}\right) \\
	& + 2\sqrt{\epsilon}\left(\atan\left(\frac{x_T}{\sqrt{\epsilon}}\right) - \atan\left(\frac{x_0}{\sqrt{\epsilon}}\right)\right) \\
	&+ 2z_0\left(\atan\left(\frac{x_T}{z_0}\right) - \atan\left(\frac{x_0}{z_0}\right)\right)
	\Bigg] \\
%	=&  \frac{B}{v}\Bigg[ x_0\ln\left(\frac{x_0^2+z_0^2}{x_0^2+\epsilon}\right) + 2\sqrt{\epsilon}\left(\frac{\pi}{2} - \atan\left(\frac{x_0}{\sqrt{\epsilon}}\right)\right) \\
%	&- 2z_0\left(\frac{\pi}{2} - \atan\left(\frac{x_0}{z_0}\right)\right) \Bigg]\\
	=& \frac{B}{v}\Bigg[ x_0\ln\left(1 + S\right) + \pi\left(\sqrt{z_0^2+Sd_0^2}-z_0\right)
	 \\
	&- 2\left(\sqrt{z_0^2+Sd_0^2}\atan\left(\frac{x_0}{\sqrt{z_0^2+Sd_0^2}}\right) - z_0\atan\left(\frac{x_0}{z_0}\right) \right)
    \Bigg].\\
	\end{align*}
%\end{proof}

\bibliography{DataLimitBib}
\bibliographystyle{ieeetr}
\end{document}

%% file: Table1.tex
\begin{table}[tb]
	\centering
	\caption[TODO]{Table of \textit{default} simulation parameters. Graph axes and labels indicate where parameters differ from defaults.}
	\begin{tabular}{|l|l|l|l|l|l|l|}
		\hline
		$B$(Hz) & $\sigma^{2}$(W) & $d_0$(m) & $G$ & $P$(W) & $v\text{(ms}^{-1}$)& $\alpha$ \\
		\hline
		$10^5$ & $10^{-8}$ & $1$ & $1$ & $1$ & $5$ & $2$ \\
		\hline
	\end{tabular}
	\label{tab:Params}
\end{table}